\documentclass[12pt]{article}
\usepackage{epsfig}
\usepackage{amssymb,graphicx}
\def\be{\begin{equation}}
\def\ee{\end{equation}}
\def\ba{\begin{array}{c}}
\def\ea{\end{array}}

\def\ben{\[}
\def\een{\]}

\newcommand{\bea}{\begin{eqnarray}}
\newcommand{\eea}{\end{eqnarray}}

\newtheorem{thm}{Theorem}
\newtheorem{cor}[thm]{Corollary}
\newtheorem{lemma}[thm]{Lemma}

\newenvironment{proof}{\noindent
 {\bf Proof.}}{\hfill$\square$\vspace{3mm}\endtrivlist}

\begin{document}

\begin{center}

{\Large \bf {

Arnold's potentials and quantum catastrophes II.

 }}

\vspace{1.3mm}

\vspace{3mm}

\begin{center}

\vspace{0.8cm}

  {\bf Miloslav Znojil}

\vspace{0.2cm}

The Czech Academy of Sciences, Nuclear Physics Institute,

 Hlavn\'{\i} 130,
250 68 \v{R}e\v{z}, Czech Republic

\vspace{0.2cm}

 and

\vspace{0.2cm}

Department of Physics, Faculty of Science, University of Hradec
Kr\'{a}lov\'{e},

Rokitansk\'{e}ho 62, 50003 Hradec Kr\'{a}lov\'{e},
 Czech Republic

\vspace{0.2cm}

{e-mail: znojil@ujf.cas.cz}

\vspace{0.3cm}

and

\vspace{0.3cm}

{\bf Denis I. Borisov}

\vspace{0.2cm}

Institute of Mathematics, Ufa Federal Research Center, RAS,
Chernyshevskii str. 112, 450008
Ufa, Russia

\vspace{0.2cm}

 and

\vspace{0.2cm}

Bashkir State University, Zaki Validi str. 32, 450076 Ufa, Russia

\vspace{0.2cm}

 and

\vspace{0.2cm}

Department of Physics, Faculty of Science, University of Hradec
Kr\'{a}lov\'{e},

Rokitansk\'{e}ho 62, 50003 Hradec Kr\'{a}lov\'{e},
 Czech Republic

\vspace{0.2cm}

 e-mail: BorisovDI@yandex.ru

\end{center}

\vspace{3mm}

\end{center}

\subsection*{Keywords:}

Schr\"{o}dinger equation; multi-barrier polynomial potentials;
avoided energy-level crossings; abrupt wavefunction
re-localizations; quantum theory of catastrophes;

\subsection*{PACS number:}

PACS 03.65.Ge - Solutions of wave equations: bound states

\newpage

\subsection*{Abstract}

In paper I (M. Znojil, Ann. Phys. 413 (2020) 168050)
it has been demonstrated that
besides the well known use of
the Arnold's one-dimensional polynomial potentials
$V_{(k)}(x)= x^{k+1} + c_1x^{k-1} + \ldots$
in the classical Thom's catastrophe theory,
some of these potentials (viz., the confining ones, with
$k=2N+1$)
could also play an analogous
role of genuine benchmark models
in quantum mechanics, especially
in the dynamical regime in which
$N+1$ valleys are separated by
$N$ barriers.
For technical reasons,
just the
ground states in the
spatially symmetric
subset of $V_{(k)}(x)=V_{(k)}(-x)$
have been considered.
In the present paper II
we will show that and how
both of these constraints
can be relaxed.
Thus, even the knowledge of the trivial
leading-order form of the
excited states
will be shown sufficient to provide
a new, truly rich level-avoiding spectral pattern.
Secondly,
the fully general asymmetric-potential scenarios will be
shown tractable perturbatively.

\newpage

\section{Introduction}

In many realistic applications of quantum mechanics
(say, in molecular or nuclear physics)
the structure of the bound-state spectrum
is often found fairly sensitive
to the variation of
parameters. In particular, after a small change
of these parameters some of
the neighboring energy levels
appear to merge and cross
(cf., e.g., our preceding paper I \cite{ArnoldI}
for references).
Although a more detailed inspection of the
spectrum reveals that the crossings
are in fact avoided \cite{vono},
the phenomenon permanently
attracts
interest
of theoreticians \cite{Berry,denis2,sdenisem}
as well as experimentalists
\cite{Stellenbosch}.

One of the reasons is that
after an analytic continuation
of the Hamiltonian in the
complex plane of a parameter,
the avoided energy-level crossing (ALC)
phenomenon
proves intimately related to
the complex singularities
called (Kato's) exceptional points
(EPs, \cite{Kato}). For example,
in the one-dimensional
Schr\"{o}dinger equation
 \be
  H(a,b,\ldots)\,\psi_n(x) = E_n(a,b,\ldots)\,\psi_n(x) \,,
  \ \ \ \ n = 0, 1, \ldots\,,
 \label{bsse}
 \ee
 \ben
 \ \ \ \ \ H(a,b,\ldots)=-\frac{\hbar^2}{2\mu}\,\frac{d^2}{dx^2}
 + V(x,a,b,\ldots)\,,\ \ \
 \psi_n(x)\,\in\,L^2(-\infty,\infty)\,
 \,
 \een
an explicit form of the ALC-EP correspondence
emerges even after the most elementary choice of the
exactly solvable harmonic-oscillator potential~\cite{ptho}.
A
message deduced from this observation
is that whenever one studies the quantum ALC phenomenon,
it makes sense
to start the analysis from analytic
potentials of polynomial form.

In paper I, for this reason, we paid attention to the study of
Schr\"{o}dinger equation (\ref{bsse}) for a specific sequence of
spatially symmetric polynomial potentials
 \be
 V(x,c_1)=x^4 +c_1\,x^2\,,\ \ \ \
 \label{spefa}
 V(x,c_1,c_2)=x^6 +c_1\,x^4+c_2\,x^2\,,\ \ \ \
 \ldots\,.
 \ee
In a way summarized in section \ref{dvojkaka} below
this enabled us to find another connection between
the quantum ALC
and the phenomenon of
bifurcation of equilibria
in certain classical dynamical systems
\cite{Thom,Zeeman}.
Indeed, sequence (\ref{spefa}) is just an
even-parity subset of family
 \be
 V_{(k)}^{(Arnold)}(x,a,b,\ldots)
 = x^{k+1} + a\,x^{k-1} + b\,x^{k-2} + \ldots\,,
 \ \ \ \ k = 1, 2, \ldots\,
 \label{pefa}
 \ee
which has been proposed by Arnold \cite{Arnoldbook}.
In an
extended version of
the Thom's theory Arnold
emphasized that
the use of  the general
polynomial potentials (\ref{pefa})
offers a geometric picture of the branching of evolutions
in the generic one-dimensional classical dynamical systems.

In paper I it has been pointed out that at the odd subscripts
$k=2N+1$ the classical catastrophes could be also paralleled, in
consistent manner, by their quantized analogues. Under the
assumption of spatial symmetry (\ref{spefa}) such a project proved
technically feasible. Via a suitable {\it ad hoc\,}
reparametrization of the couplings in (\ref{pefa}), a systematic
classification of the sudden ALC-related changes of the state of the
systems has been obtained. Naturally, due to the characteristic
mathematical fact of the tunneling, not all of the classical types
of the catastrophe found their fully analogous quantum counterpart.
Some of them (like, e.g., cusp) got smeared out and smoothed after
quantization. What still survived were the abrupt, catastrophe-like
changes of the topological properties of probability densities. They
remained measurable and, hence, phenomenologically meaningful. For
this reason it has been suggested to call these specific types of
the ALC-related bifurcations of equilibria ``quantum relocalization
catastrophes'' (QRC).

In our present paper we intend to
complement these results
by a few vital and important addenda.
First of all we will turn attention to
the role of asymmetry in section \ref{roas}
and to its possible influence upon the
quantum-catastrophic scenarios (section \ref{trojka}).
Besides such an extension of
physics
we will reconsider also
the mathematical
problem of the growth of the technical
obstacles encountered after one removes
the parity-symmetry constraint $V(x)=V(-x)$.
Resolution of the challenge will be shown provided by
perturbation theory.
We will decompose the general Arnold's confining
quantum
potential (\ref{ourpo})
into a dominant, even-parity unperturbed component $V^{(even)}(x)$
complemented by an
odd-parity perturbation $V^{(odd)}(x)$.

The basic idea will lie in a combination of the standard assumption
of the smallness of perturbations, $V^{(odd)}(x)={\cal
O}(\epsilon)$, with an innovative ``re-use'' of the above-mentioned
technical user-friendliness of the reparametrizations of the
even-parity polynomials. Such a reparametrization of coefficients
proved essential, in paper I, during the localizations of ALCs for
the even Arnold's polynomial potentials
$V^{(even)}(x)=V^{(even)}_{(N)}(x)$ of degree $2N+2$. Here, the
trick will be re-applied also to the products
$V^{(odd)}(x)=\epsilon\,x\,V^{(even)}_{(M)}(x)$. The control of the
shape of a perturbation in $V^{(odd)}(x)$ will be transferred to the
control of the shape of the even-parity auxiliary polynomial of a
smaller degree $2M+2 < 2N+2$. The repeated application of the same
reparametrization trick to perturbations of an independently tunable
shape will be illustrated in section \ref{petka}.

In the final part of the paper we will
turn attention to excited states.
We will oppose the popular opinion
(advocated even in paper I) that
the only meaningful form of quantization would be a
transition from the study of minima of the potentials
(representing the stable long-term equilibria of classical systems)
to the
analogous analysis of the quantum ground states.
We will remind the readers that such an attitude has recently been
criticized, by Smilga \cite{smilgasem}, as over-restrictive.
He proposed that in any
phenomenologically responsible analysis of quantum systems
the excited states
may often be treated as long-term stable structures.
We only have to distinguish between
the ``malign'' (i.e., the de-excitation supporting)
and ``benign''   (i.e., the de-excitation not
supporting) nature of the states with respect to the
relevant physical
perturbations.
This means
that the
stability status of the quantum excited states is
model-dependent. The conventional emphasis upon
the destructive role of perturbations should not be exaggerated,
especially in the most common case of
the low-lying spectra of bound states
which are often, in real world, all long-time stable.
In this spirit we will demonstrate, in section \ref{ctyrka},
the technical feasibility of an
extension of the
ALC-related
constructions from the
exclusive ground-state scenarios of Ref.~\cite{ArnoldI}
to the whole
low-lying spectra of the excited bound states.

In a slightly broader perspective some open questions and
qualitative aspects of our present results will be discussed in
Appendices A and B, with a final summary presented in section
\ref{soummar}.


\section{The smearing caused by the tunneling\label{dvojkaka}}

\subsection{Multiple-well potentials}

In phenomenological applications the choice of the potential in
Schr\"{o}dinger equation (\ref{bsse}) is usually dictated by an
expected correspondence between quantum system and its suitable
classical-physics analogue. Such a purely intuitive idea played an
important role not only during the birth of quantum mechanics but
also in the various related methodical considerations. In paper I,
in particular, the classical-quantum correspondence has been found
inspiring due to its relevance in dynamical systems. On the
classical-physics side, the key role has been known to be played by
the singular evolution scenarios called catastrophes \cite{Zeeman}
and/or bifurcations \cite{Arnoldium}. On the quantum theory side,
unfortunately, the singularities of such a type are usually assumed
``smeared out'' by the quantization \cite{Mukhanov}. For this reason
the progress in this field is far from rapid \cite{catast}.

Incidentally, the
range of applicability of the
Thom's
qualitative picture of
evolution patterns
was by far not restricted to physics.
In fact, the basic motivation
of the development of the
catastrophe-related mathematics lied, originally,
in biology \cite{Thom}. Moreover,
its intuitive geometric form found, paradoxically,
various formal and descriptive applications not only in
disciplines like theoretical chemistry \cite{Thomchem}
but even in
quantum
physics itself \cite{catanucl,dell1,dell2}.
Obviously, the deep conceptual differences between the
classical and quantum pictures of reality are reflected also
in the underlying mathematics. The necessity of replacement of the
observed numbers (representing, say, an energy or position
of a classical point particle) by the operators in Hilbert space
implies that the traditional variational and/or geometric
tractability of these observables is lost,
and the tools of the functional analysis have to be used
\cite{catast}.
Even the
Arnold's polynomial potentials (\ref{pefa})
will lose their quantum-mechanical applicability
at the odd subscripts $k=2N+1$. Thus, we have to
restrict our attention to the subset of polynomials
 \be
 V(x)=x^{2N+2}+a\,x^{2N}+b\,x^{2N-1}+\ldots +z\,x\,
 \label{ourpo}
 \ee
where an elementary shift of the
origin on the real line of $x$ enabled us to eliminate
the subdominant term.
Without any loss of generality we were also able to
fix $V(0)$ via a suitable
shift of the energy scale.

\subsection{Double-well potentials ($N=1$)}

In the framework of the classical theory of catastrophes
our
confining $2N-$parametric
potentials (\ref{ourpo}) can be characterized,
according to Arnold \cite{Arnold},
by the Dynkin diagrams $A_k$
with odd $k=2N+1$.
After quantization, i.e., after insertion in
Schr\"{o}dinger equation (\ref{bsse})
one could have a tendency to start from
the two-parametric $N=1$
model representing a {\em classical\,} catastrophic
scenario called
cusp,
 \be
 V(x) =V^{(cusp)}(a,b,x)
 ={x}^{4}+a\,{x}^{2}+b\,x
 \,.
  \label{[2]}
 \ee
One reveals that also
such a choice is to be excluded because,
in contrast to
the classical case,
the solutions
change smoothly with the parameters.
The related classical cusp catastrophe
proves really smeared out by the quantization \cite{ArnoldI}.
Due to the quantum tunneling the
wave function of the
system will reside, even in double-well case,
in both of the valleys.
Hence, the change of the sign of $a$
will not induce any abrupt changes of the observable
features of the system.

After we
incorporate the second variable coupling $b$,
the left-right symmetry becomes broken, but
the smoothness of the dynamics remains
unchanged.
The change of $b$ merely
moves one of the minima downwards.
Locally, the well in its vicinity gets broadened so that
irrespectively of the sign of $b$
the other minimum moves up and,
from the point of view of the low-lying spectrum, it
loses its relevance.
The ground-state wave function
$\psi_0(x)$
will tunnel out of the upper well.
The process remains smooth. In a search for
abrupt changes
our attention has to be redirected to
the less elementary shapes of
potentials (\ref{ourpo}) with $N \geq 2$.

\subsection{Reparametrizations  ($N=3$)}

At the larger $N$, the reparametrization-based
localization of the ALC catastrophes
is vital but it
could already be perceived as a challenging task.
This has been emphasized in paper I.
Still, we intend to demonstrate that in a way sampled at $N=3$,
an exhaustive study of the menu of catastrophes
generated by the asymmetric, six-parametric Arnold's potential
 \be
 V^{(N=3)}(x)
 = x^8 +a\,x^6 +b\,x^5  +c\,x^4 +d\,x^3  +f\,x^2 +g\,x\,
 \label{butter3}
 \ee
can be simplified. What is only needed is
a reinterpretation and restriction of the
general asymmetric potential
to a weakly asymmetric
model defined as a perturbation
of the symmetric three-parametric unperturbed polynomial
 \be
 V_0(x)
 = x^8 +a\,x^6   +c\,x^4   +f\,x^2 \,.
 \label{butter3}
 \ee
The latter model can be reparametrized:
we
first evaluate its derivative and factorize it in
a way which is slkightly different from the recipe of Ref.~\cite{ArnoldI},
 $$
 V_0'(x)\sim \left( {x}^{2}-{\alpha}^{2} \right)  \left( {x}^{2}-{\alpha}^{2}-
 {\beta}^{2} \right)  \left( {x}^{2}-{\alpha}^{2}-{\beta}^{2}-{
 \gamma}^{2} \right)\,.
 $$
This
enables us to arrive at the marginally simpler reparametrizations
of the couplings,
 $$
 a=-4\,{\alpha}^{2}-8/3\,{\beta}^{2}-4/3\,{\gamma}^{2}\,,
 $$
 $$
 c= 8\,{\alpha}^{2}{\beta}^{2}
 +4\,{\alpha}^{2}{\gamma}^{2}+2\,{\beta}^{4}
 +6\,{{\alpha}}^{4}+2\,{\beta}^{2}{\gamma}^{2}\,,
 $$
 $$
 f=-4\,{\alpha}^{2}{\beta}^{2}{\gamma}^{2}
 -4\,{\alpha}^{6}-8\,{\alpha}^{4}{\beta}^{2}
 -4\,{\alpha}^{4}{\gamma}^{2}-4\,{\alpha}^{2}{\beta}^{4}\,.
 $$
In a subinterval of coordinates which lie closer to the origin
the shape of the potential
reaches the two symmetric minimal values
 $$
 V_{(\rm inner\ minimum)}=
 -{\alpha}^{8}-8/3\,{\alpha}^{6}{\beta}^{2}-4/3\,{\alpha}^{6}{
 \gamma}^{2}-2\,{\alpha}^{4}{\beta}^{4}
 -2\,{\alpha}^{4}{\beta}^{2}{\gamma}^{2}\,.
 $$
At these coordinates, by definition, the first derivative
of the potential vanishes while
 $$
 V_{(\rm inner\ minimum)}^{\rm (second \ derivative)}=
 16\,{\alpha}^{2}{\beta}^{4}+16\,{\alpha}^{2}{\beta}^{2}{\gamma}^{2}
 $$
is positive (and large) so that
the harmonic-oscillator approximation is validated.

At the two outer minima the value of the
potential has the slightly more complicated form
 $$
 V_{(\rm outer\ minimum)}=
 -{\alpha}^{8}-2\,{\alpha}^{4}{\beta}^{2}{\gamma}^{2}+1/3\,
 {\beta}^{8}-2/3\,{\beta}^{2}{\gamma}^{6}-8/3\,{\alpha}^{6}
 {\beta}^{2}-4/3\,{\alpha}^{6}{\gamma}^{2}-2\,{\alpha}^{4}
 {\beta}^{4}+
 $$
 $$
 +2/3\,{\beta}^{6}{\gamma}^{2}-1/3\,{\gamma}^{8}
 $$
while the value of the second derivative is positive as it should be,
 $$
 V_{(\rm outer\ minimum)}^{\rm (second \ derivative)}=
 16\,{\beta}^{4}{\gamma}^{2}+16\,{\alpha}^{2}{\beta}^{2}
 {\gamma}^{2}+32\,{\beta}^{2}{\gamma}^{4}+16\,{\gamma}^{6}+16\,{
 \alpha}^{2}{\gamma}^{4}
 $$
A comparison of these formulae with their rescaled analogues in paper I
reveals that the absence of the rescaling keeps in fact
the transparency of the formulae practically unchanged.
This means that also the evaluation of the locally supported low-lying spectra
as well as the inclusion of asymmetric perturbations remains to be
an entirely routine exercise
which can be left to the readers.

%
%

\section{The role of the asymmetry of $V(x)$\label{roas}}

%

Let us
consider the
general confining Arnold's potential
 \be
 V_{(2N+1)}^{(Arnold)}(x,a,b,\ldots,q)
 = x^{2N+2} + a\,x^{2N} + b\,x^{2N-1} + \ldots +p\,x^2 +q\,x
 \label{genai}
 \ee
and
let us simplify the discussion by
the special choice of $q=0$ and $b>0$.
In any nontrivial case this enables us to
introduce the even and odd components of the potential,
 \be
 V^{(even)}_{(N)}(x)=\frac{1}{2}\,
 \left [V^{(Arnold)}_{(2N+1)}(x)+ V^{(Arnold)}_{(2N+1)}(-x)\right ]
 = x^{2N+2} + a\,x^{2N} + c\,x^{2N-2} + \ldots+ p\,x^2\,,
 \label{faks}
 \ee
 \be
 V^{(odd)}_{(N)}(x)=\frac{1}{2}\,
 \left [V^{(Arnold)}_{(2N+1)}(x)- V^{(Arnold)}_{(2N+1)}(-x)\right ]
  = b\,x \left [x^{2M+2}+ a'\,x^{2M}
  + c'\,x^{2M-2} + \ldots +p'\,x^2 \right ]
 \,
 \label{fakl}
 \ee
where $M=M(N)=N-2$ and $a'=d/b$ and $c'=f/b$, etc.

\subsection{Spatially symmetric special cases}

Our analysis
of the related ALC/QRC phenomena will be separated in two halves.
In its first half
we will study just the models in which the original potential
is even,
$V^{(Arnold)}_{(2N+1)}(x)=V^{(even)}_{(N)}(x)$.
Along the lines indicated above
we will reparametrize the $N-$plet of its
coupling constants $a,c,\ldots,p\ $ in terms of another $N-$plet of
parameters $\alpha,\beta,\ldots,\omega$ which specify the
spatially symmetric $N-$plet of zeros
$\xi_n$
of the first derivative of the potential,
 \be
  \left [V^{(even)}_{(N)}\right ]'(x) =(2N+2)\, x
 \left( {x}^{2}-{\alpha}^{2} \right)\,  \left( {x}^{2}-{\alpha}^{2}-
 {\beta}^{2} \right) \ldots
 \left( {x}^{2}-{\alpha}^{2}-{\beta}^{2}-\ldots -{
 \omega}^{2} \right)\,.
 \label{fakkd}
 \ee
We will assume that
all of the new parameters $\alpha,\beta,\ldots,\omega$
are large, ${\cal O}(\lambda)$, $\lambda \gg 1$, making
our spatially symmetric potential composed of an
$(N+1)-$plet of the well separated deep wells.
Near the local minima, these wells may be represented,
with good precision,
by the exactly solvable harmonic-oscillator potentials,
 \be
 V(x) \sim F_n+G_n(x-X_n)^2 \
 + \ {\rm corrections}\,,
 \ \ \ \ x-X_n=
  {\rm small}\,,
 \ \ \ \ G_n>0\,.
 \label{apro}
 \ee
This observation, thoroughly analyzed in paper I,
can be given the more explicit form.

\begin{lemma}\cite{ArnoldI}
\label{lemma1}
In the dynamical regime with the large parameters
$\alpha= {\cal O}(\lambda),\beta= {\cal O}(\lambda),\ldots$
and
positions of minima $X_n = {\cal O}(\lambda^{})$,
$\lambda \gg 1$, the
coefficients in formula (\ref{apro})
will be large,
 \be
 F_n=F_n(\alpha,\beta,\ldots) = {\cal O}(\lambda^{2N+2})\,,\ \ \ \
 G_n=G_n(\alpha,\beta,\ldots) = {\cal O}(\lambda^{2N})\,,\ \ \ \
 n=0,1,\ldots,N\,.
 \label{ormag}
 \ee
\end{lemma}

This is an important result because it excludes the possibility of
the loss of the applicability of the present ubiquitous
leading-order harmonic-oscillator approximations. In other words,
the Lemma guarantees that the local minima of the Arnold's potential
remain ``pronounced'' and will not ``flatten'' even in the small
vicinity of any Kato's complex EP degeneracy, i.e., in our present terminology,
directly in the
``catastrophic'' real-coupling ALC/QRC dynamical regime.
In parallel, the barriers between the separate minima remain,
for the same reason,
high and thick. In the light of paper I, this in fact leads to the most
important consequence concerning the wave functions: during their
``passage through the barrier'',
their decrease is
exponentially quick.

Out of the multi-well ALC-supporting ``bifurcation-admitting''
parametric domain (the size of
which is exponentially small), the
numerical value of the wave function becomes negligible
near any non-dominant minimum of the potential.
Interested readers are advised to simulate and check the phenomenon
(and, in particular, the speed of the exponential suppression
of the size of the wave function inside a high and thick barrier)
via an exactly solvable square-well model, with the typical
closed-form wave functions
sampled via example Nr. 12 on page 42 in monograph \cite{Const}).

In an immediate consequence of Lemma \ref{lemma1}, also
the values of the lowermost ground-state energies
 \be
 E_0^{(n)}(\alpha,\beta,\ldots)=F_n(\alpha,\beta,\ldots)
 +\sqrt{G_n(\alpha,\beta,\ldots)}+ \ {\rm corrections}\,,
 \ \ \ n=0,1,\ldots,N
 \label{esti}
 \ee
will be large. This observation is certainly useful for a
generalization of the $N=2$ ALC/QRC condition (see
Eq.~(\ref{grstsr}) below) to any $N$.

Thirdly, in the dynamical regime with the large parameters $\alpha=
{\cal O}(\lambda),\beta= {\cal O}(\lambda),\ldots$ the following set
of the ALC/QRC conditions
 \be
 E_0^{(m)}(\alpha,\beta,\ldots)=E_0^{(n)}(\alpha,\beta,\ldots)
 \,,
 \ \ \ \ \ m>n=0,1,\ldots,N\,
 \label{coine}
 \ee
is, due to the symmetry of the potential, over-complete. This is an
apparent puzzle which has the following remedy.
\begin{cor}
For both even $N=2J$ and odd $N=2J+1$ the ``thin'' sub-domain ${\cal
D}^{(QRC)}_{(maximal)}$ of the parameters $\alpha,\beta,\ldots$
which would be compatible with a complete ALC degeneracy of the
ground-state energies (\ref{esti}) will be determined, up to small
uncertainties, by any subset of $J$ independent constraints
(\ref{coine}). \label{2co}
\end{cor}
By constraints (\ref{coine}), in general,
the initial $N-$plet
$\alpha,\beta,\ldots,\omega$
becomes reduced, roughly, to one half (more precisely, to
${\cal N}=entier[(N+1)/2]$
parameters).
An approximate ALC coincidence of {\em all\,} of the
local-well ground-state energies will be achieved.
For illustration we may recall
the most elementary $N=2$ sample ${\cal D}^{(QRC)}_{(maximal)}
=\{(\alpha,\beta)|\beta=\sqrt{2} \alpha \pm \rm small\ corrections\}$
of
the ``thin'' domain of Corollary \ref{2co} above.

At any $N$ the ALC-related measurable probability density
$\varrho(x)=\psi^*(x)\psi(x)$
will be
spread over all of the vicinities of the minima $X_n$.
The corresponding QRC equilibrium
characterized by the almost fully degenerate spectrum is highly unstable
of course.
Even the smallest
perturbation
will
move the
parameters out of the equilibrium domain ${\cal D}^{(QRC)}$,
causing a collapse and opening a transition, presumably,
to a new, stable (or
at least less unstable) equilibrium.

\subsection{Slightly asymmetric potentials}

In the preceding paragraph
we implicitly
assumed that the perturbations
causing an unfolding of the maximal
ALC/QRC collapse do not violate the spatial symmetry of
the potential. In such an arrangement, the control of the
relocalizations may be expected to
proceed, at arbitrary $N>2$,
along the lines which would parallel the same pattern of
explicit predictions.
This means that the changes of the ``outer'' parameters
would
influence, predominantly, the shifts
of the ``outer'' local-well minima and/or of the associated
low-lying ``outer'' subspectra.

No similar intuitive {\it a priori\,} estimate of the
behavior of the subspectra
exists in the full-fledged asymmetric models with $q=0$.
At any $N$,
the analysis can still be simplified
by the decomposition
 \be
 V_{(2N+1)}^{(Arnold)}(x,a,b,\ldots,p)= V^{(even)}_{(N)}(x)+
  V^{(odd)}_{(N)}(x)\,
 \label{gena}
 \ee
of the general $N \geq 2$ Arnold's potential (\ref{genai}) with $q=0$.
Thus, let us assume that the even component $V^{(even)}_{(N)}(x)$
of the potential
becomes fixed and fine-tuned to the above-described,
maximally ALC unstable equilibrium considered, for the sake of
definiteness, just in the ground-state regime. Then
we may expect that
a rich menu of the
QRC transitions to an alternative stable equilibrium will be
realizable via the odd component of the potential
with a sufficiently small $b=\epsilon$,
 $$
 V^{(odd)}_{(N)}(x)=\epsilon\,x\, V^{(even)}_{(N-2)}(x)
 $$
or with $b=0$ and with a sufficiently small $d=\epsilon$,
 $$
 V^{(odd)}_{(N)}(x)=\epsilon\,x\, V^{(even)}_{(N-4)}(x)\,,
 $$
etc.
Along these lines, the trick is that we may ask for the
presence
of pronounced minima and maxima not only
in the dominant even-parity potential $V^{(even)}_{(N)}(x)$
but also in the small, ${\cal O}(\epsilon)$ contributions
coming from the odd
function correction $V^{(odd)}_{(N)}(x)=
\epsilon\,x\,V^{(even)}_{(M)}(x)$ with $M=N-2$ or $M=N-4$, etc.
Once such a small ${\cal O}(\epsilon)$ term is
added to the
dominant but sensitive, ALC-fine-tuned $V^{(even)}_{(N)}(x)$
of Eq.~(\ref{apro}), it may cause
enhancements and/or suppressions of the
dominant-potential wells
via the new, movable ${\cal O}(\epsilon)$
wells and/or barriers in $V^{(odd)}_{(M)}(x)$.
This might facilitate the selection and control of
the ultimate QRC equilibria as long as
the coordinates $x=X_n=X\,\lambda>0$
of the local minima as well as
the large parameters $F_n=F\,\lambda^{2N+2}$ and $G_n=G\,\lambda^{2N}$
in Eq.~(\ref{apro})
are of the order of magnitude as specified in Lemma \ref{lemma1}.

One can summarize that our task is reduced to the
analysis of the adaptability of the properties
of the spatially antisymmetric perturbation component
$\epsilon\,x\,V^{(even)}_{(M)}(x)$
of the full potential.
It is important that in such a component the
maximal power $x^{2M+3}$
can be much smaller
than the
maximal power $x^{2N+2}$ characterizing the full, perturbed
Arnold's potential.
Hence, the generic multi-well
form of the perturbation
can be kept as simple as possible.
In particular, its freely variable
local extremes could be more easily matched, added to, or subtracted from, their
unperturbed partners.
Decisively, such a matching
can be facilitated by the following observation.

\begin{lemma}
\label{lemma3}
In a vicinity of the minimum $\lambda\,X>0$ of
$V(x) = \lambda^{2M+2}\,F + \lambda^{2M}\,G\,(x-\lambda\,X)^2$
with  $G>0$,
the minimum of the third-order polynomial
$x\,V(x)$  in $x$
lies at $x_0=\lambda\,(1+\delta)\,X\,$ where the shift
is $\lambda-$independent,
 \be
 \delta=-\frac{F}{G\,X^2+X\,\sqrt{G^2\,X^2-3\,F\,G}}\,.
 \ee
\end{lemma}
\begin{proof}
The odd function $x\,V(x)$ of $x$ has two local extremes (viz., a local
maximum and a local minimum), with the coordinates given as roots of
quadratic equation.
\end{proof}

This result shows that
in a small vicinity of a preselected minimum
of a dominant symmetric
part of a slightly asymmetric Arnold's potential,
an enhancement or suppression of this minimum can be
achieved via an
antisymmetric {\it ad hoc\,} perturbation
$V^{(odd)}_{(N)}(x)$, the global shape of which
can be controlled by an $M-$plet of independent parameters
where $M$ can be chosen much smaller than $N$.


\section{Quantum catastrophes\label{trojka}}

\subsection{Symmetric triple-well potentials}

Schr\"{o}dinger equation
 \be
 \left[- {\Lambda}^2\,\frac{d^2}{dx^2} +
 V^{(butter\!fly))}(x)
 \right ] \psi_n (x) =E_n\,
 \psi_n (x)
  \,, \ \ \ \ \
 {\Lambda}^2={\hbar^2}/({2\mu})\,, \ \ \ \ \
 n=0, 1, \ldots\,
 \label{6seh}
 \ee
with the four-parametric potential
 \be
 V^{(butter\!fly)}(x)
 = x^6 +a\,x^4 +b\,x^3  +c\,x^2 +d\,x\,
 \label{butterb1}
 \ee
is simpler to solve in its spatially symmetric
version where $b=d=0$.
The basic
technical aspects of such a
special case
were discussed in paper I.
In a slightly different notation
let us now summarize these results briefly.
First, in the spirit of {\it loc. cit.},
potential
 \be
 V(x)
 = x^6 +a\,x^4 +c\,x^2 \,
 \label{butter00}
 \ee
will only be considered here in its
most interesting
deep-triple-well dynamical regime.
Under this assumption,
the first derivative
of the potential can be factorized
in terms of
the coordinates of the extremes of $V(x)$
or, in other words,
in terms of suitable real $\alpha$ and $\beta$,
 \be
 V'(x)
 = 6\,x^5 +4\,a\,x^3 +2\,c\,x=
 6\,x\,(x^4+\frac{4\,a}{6}\,x^2+\frac{2\,c}{6}\,x
 )=6\,x\,(x^2-\alpha^2)\,(x^2-\alpha^2-\beta^2)
  \,,
 \label{derbutter00}
 \ee
This induces a reparametrization of the original couplings,
 \be
 a=a(\alpha,\beta)=-{3}\,(\alpha^2+\beta^2/2)\,,\ \ \ \
 c=c(\alpha,\beta)=3\,\alpha^2\,(\alpha^2+\beta^2)\,.
 \label{[10]}
 \ee
The pronounced, deeply triple-well
shape of the potential
possessing
two high and thick inner barriers
will be achieved when
and only when the coordinates of the extremes are chosen
sufficiently large,
$\alpha^2 \gg {\Lambda}^2$ and $\beta^2 \gg {\Lambda}^2$.
Thus, in units
such that ${\Lambda}^2=1$ we have to have
$\alpha \gg 1$ and $\beta \gg 1$.

\subsection{Ground states and the avoided level crossing phenomenon}

The latter, phenomenologically motivated assumption
simplifies the approximate construction of bound states.
In the case of the
dominance of the central attraction
the approximate low-lying
spectrum
will acquire the elementary
leading-order form
 \be
 E_n^{\rm (central)}=(2n+1)\,\sqrt{c(\alpha,\beta)}
 + {\rm higher\ order\  corrections}\,,
 \ \ \ n=0, 1, \ldots\,.
 \label{centrsp}
 \ee
The alternative assumption of the
dominance of the off-central attraction
yields the almost degenerate
energy-level doublets
 \be
 E_m^{\rm (even/odd)}=V(\sqrt{\alpha^2+\beta^2})
 +(2m+1)\,\Omega + {\rm higher\ order\  corrections}\,,
 \ \ \ m=0, 1, \ldots\,,
  \label{offcent}
 \ee
 $$
 \Omega=\sqrt{V''(\sqrt{\alpha^2+\beta^2})/2}\,,\ \ \
 V''(\sqrt{\alpha^2+\beta^2})=
12\,\alpha^2\,\beta^2+12\,{\beta}^{4}\,$$
which are, naturally, never
strictly degenerate.

As long as we have
 $
 V(\sqrt{\alpha^2+\beta^2})={\alpha}^{6}
 +3/2\,{\alpha}^{4}\beta^2-1/2\,{\beta}^{6}
 $
the ground-state
special case of formula (\ref{offcent}) reads
 \be
 E_0^{\rm (even/odd)}={\alpha}^{6}
 +3/2\,{\alpha}^{4}\beta^2-1/2\,{\beta}^{6}
 +\sqrt{6\,\alpha^2\,\beta^2+6\,{\beta}^{4}}
 + {\rm   corrections}\,.
 \label{offfin}
 \ee
We may
set $\beta^2=\mu^2\,\alpha^2$
with, say, $\mu={\cal O}(1)$. Asymptotically
(i.e., in the regime of very large parameters
$\alpha \gg 1$ and $\beta \gg 1$)
we then deduce,
from formula (\ref{centrsp}),  that
$E_n^{\rm (central)}={\cal O}(\alpha^2)$.
In contrast,
Eq.~(\ref{offfin}) implies
that $E_m^{\rm (even/odd)}={\cal O}(\alpha^6)$.
This comparison reveals the clearly dominant behavior of the
off-central minimum of the potential.

In the leading-order approximation the ``quantum-catastrophic''
instant of
transition between the
central and off-central dominance
of the probability density of the ground states will be
merely
dictated by the trivial constraint
$V(\sqrt{\alpha^2+\beta^2}) ={\cal O}(\alpha^2)$,
i.e., up to corrections, by the elementary equation
$V(\sqrt{\alpha^2+\beta^2}) =0$.
This enables us to
deduce that the quantum relocalization catastrophe
becomes controlled by $\beta$ and
realized at
$\mu^2\approx \mu_\infty=2$, i.e.,
along the line of $\beta \approx \sqrt{2}\,\alpha$.
Only in the next order approximation one has to impose
the explicit
ALC condition
 \be
 E_0^{\rm (central)}(\alpha,\beta) = E_0^{\rm (even/odd)}(\alpha,\beta)
 + \rm small \ corrections\,.
 \label{grstsr}
 \ee
After an appropriate change of the
notation conventions,
a more detailed numerical
analysis of its solutions may be found
described in paper I.


\section{Asymmetric $N=2$ illustration\label{petka}}

The turn of attention to the spatially asymmetric general version
of the butterfly \cite{Thom} potential
 \be
 V(x)
 = x^6 +a\,x^4 +b\,x^3 +c\,x^2+d\,x\,
 \label{butter00}
 \ee
opens a number of possibilities of reaching an asymmetric
quantum catastrophe.
A small change of parameters could now
cause an abrupt jump of the
dominant part of the observable
probability density (i.e., of
function $\rho(x)=\psi^*(x)\psi(x)$)
between the central and {\em strictly one\,} of the
off-central local minima.

\subsection{Simplification of mathematics: $d=0$}

For illustration, for the sake of brevity,
let us set $d=0$ and keep
the asymmetry
controlled by the single coupling constant $b \neq 0$.
In the limit of $ b \to 0$ the left-right-symmetric set
of the zeros
of polynomial (\ref{derbutter00})
 \be
 \{ -\sqrt{\alpha^2+\beta^2},-\alpha,0,\alpha,
 +\sqrt{\alpha^2+\beta^2} \}\,
 \label{theset}
 \ee
is known. Thus, for our present methodical purposes it will suffice
to keep the asymmetry
small. Using,
without any loss of generality, a positive $b=\epsilon > 0$,
the set of extremes (\ref{theset}) becomes modified,
in general, as follows,
 $$
 \{ -\sqrt{\alpha^2+\beta^2}-\epsilon\,p^2(\epsilon),-\alpha
 +\epsilon\,q^2(\epsilon),0,\alpha+\epsilon\,u^2(\epsilon),
 \sqrt{\alpha^2+\beta^2}-\epsilon\,v^2(\epsilon) \}\,.
 $$
The weakly $\epsilon-$dependent shifts
$p^2(\epsilon)$, $q^2(\epsilon)$, $u^2(\epsilon)$
and $v^2(\epsilon)$ must be all
positive. Their values can easily be
determined from their definition $V'(x)=0$,
i.e., from equation
 \be
 V'(x)= 6\,x^5 +4\,a\,x^3 +3\,\epsilon\,x^2+2\,c\,x
 =6\,x\,(x^2-\alpha^2)\,(x^2-\alpha^2-\beta^2)+3\,\epsilon\,x^2=0
  \,.
 \label{derbut1}
 \ee
Step by step, the
insertion of the first-order
ansatz $x=\alpha + u^2\epsilon$ using
abbreviation $u^2=u^2(0)$
in Eq.~(\ref{derbut1}) will yield the relation
 \be
 V'(x)
  \sim
 2\,(x^2-\alpha^2)\,(x^2-\alpha^2-\beta^2)+\epsilon\,x
 \sim -4\,u^2\,\,
 \beta^2
 +1
 ={\cal O}(\epsilon)
  \,.
 \label{derbut1u}
 \ee
Its leading-order solution is $u^2=1/(4\,\beta^2)$. Similarly,
the
insertion of ansatz $x=\sqrt{\alpha^2+\beta^2}-v^2\epsilon\,$
in the same equation yields an analogous
solution $v^2=1/(4\,\beta^2)$.
Along the same lines
we also obtain $q^2=1/(4\,\beta^2)$ and, finally,
$p^2=1/(4\,\beta^2)$. If needed, a systematic
evaluation of the next-order corrections would be
lengthier but also routine,
yielding
 $$
 u^2(\epsilon)=
 \frac{1}{4\,\beta^2}
 +\left ({\frac {{\beta}^{2}+4\,{\alpha}^{2}}
 {32\,\alpha\,{\beta}^{6}}}
 \right)\,\epsilon + {\cal O}(\epsilon^2)
 $$
etc.

\subsection{Asymmetry treated as a perturbation}

We are now prepared to analyze the above-mentioned asymmetric
relocalization catastrophe as a scenario during which, after a small
change of parameters, the mechanism of quantum tunneling would force
the dominant part of the observable probability density  $\rho(x)$
to perform an abrupt jump between the central and the leftmost local
minimum of the potential. In other words, the initial localization
of the dominant part of $\rho(x)$ in the center will get suppressed
while the initially negligible component of $\rho(x)$ near the left
minimum will become abruptly dominant.

In a way explained in paper I, the explicit demonstration that the
probability density indeed goes through a sharp change at the
catastrophe is based on the very initial formulation of the problem
in which the wave functions of the system are well approximated by
the harmonic-oscillator states near all of the local minima. The
task of an approximate quantitative description of such a jump is
then feasible, simplified further by the fact that the inclusion of
asymmetry does not modify the central candidate (\ref{centrsp}) for
the low-lying spectrum.

The construction of its left-well alternative is also not too
complicated because its dominant component is represented just by
the deep and $\alpha-$sensitive ${\cal O}(\alpha^6)$ minimum of the
potential. Its asymmetry-dependence {\it alias\,}
$\epsilon-$dependence can be evaluated in its closed leading-order
form to read
 \be
 V(-\sqrt{\alpha^2+\beta^2}-\epsilon\,p^2(0) )
 = V(-\sqrt{\alpha^2+\beta^2})
 -\epsilon\,(\sqrt{\alpha^2+\beta^2})^3
 +{\cal O}(\alpha^2) +{\cal O}(\epsilon^2)\,.
 \label{levice}
 \ee
Thus, the only rather lengthy calculation is needed for the
evaluation of the level-spacing parameter corresponding to the
leftmost local well. One obtains the formula
 $$
 V''(-\sqrt{\alpha^2+\beta^2}-\epsilon\,p^2(\epsilon) )
 = V''(-\sqrt{\alpha^2+\beta^2})
 +
 3\,{\frac {\sqrt {{\alpha}^{2}+{\beta}^{2}}
 \left( 4\,{\alpha}^ {2}+5\,{\beta}^{2} \right) }
 {{\beta}^{2}}}\,\epsilon+{\cal O}(\epsilon^2)\,
 $$
which implies that the ${\cal O}(\epsilon)$ correction to the
level-spacing parameter $\Omega$ is of order ${\cal O}(\alpha)$,
i.e., inessential (cf. Eq.~(\ref{offcent}) above). Only the first
two terms of Eq.~(\ref{levice}) remain relevant for the proof of the
existence of the asymmetric relocalization catastrophe.

Its explicit approximate localization proceeds via relation
 \be
 {\alpha}^{6}+3/2\,{\alpha}^{4}{\beta}^{2}-1/2\,{\beta}^{6}
 =({\alpha^2+\beta^2})^{3/2}\,\epsilon\, + \rm corrections
 \label{tare}
 \ee
i.e.,
 \be
 \epsilon=-{\frac{1}{ 2}\,{\alpha}^{3}\,\delta} \,\sqrt{ 3+\delta } \
 + \rm corrections\,.
 \label{cubi}
 \ee
Here, the (small) value of $\epsilon$ and the (large) value of
$\alpha$ should be interpreted as an input information about the
potential (i.e., about dynamics). Once we return to the ansatz
$\beta^2=\mu^2\,\alpha^2$ with a (small) variable $\delta$ in
$\mu^2=2 +\delta$ we obtain the desirable solution of the
catastrophe-determining Eq.~(\ref{tare}) in its linearized
leading-order form
 \be
 \delta=-\frac{2\,\epsilon}{\sqrt{3}\,\alpha^3}\,.
 \label{tosol}
 \ee
This formula indicates that the impact of the asymmetry
of the potential
is suppressed by the factor of $\alpha^{-3}$.
Still, in the light of Eq.~(\ref{[10]}) where
$a={\cal O}(\alpha^2)$ and $c={\cal O}(\alpha^4)$, an
asymptotic
enhancement of the magnitude of the symmetry-violating
couplings up to $b=\epsilon={\cal O}(\alpha^3)$
seems appropriate.

In such an extended dynamical regime,
a return to the more precise
implicit cubic-equation definition (\ref{cubi})
of the (negative) critical shift $\delta={\cal O}(1)$
would be necessary.
The resulting, more visible
asymmetric quantum relocalization catastrophe
will be
characterized by
a substantial, non-negligible decrease
of the critical value of the
ratio $\mu^2=\beta^2/\alpha^2$ of the two
dynamical parameters of the model.

\section{Excited states ($N=2$)\label{ctyrka}}

The above-discussed conclusions
concerning the
existence of the ground-state
relocalization transitions
in the spatially symmetric potentials
reconfirm the results
of paper I.
We only have to keep in mind here that
in {\it loc. cit.} the meaning of some of the symbols
was different.
Irrespectively of that, in both cases
one works with large $\alpha \gg 1$ and $\beta \gg 1$
so that the higher-order
anharmonicities remain small.
The
pronounced form of the minima and maxima
of the potential
is guaranteed. This
opens the way towards an innovation in which
the new
``quantum-catastrophic''
relocalization transitions involve also
the
excited states.

In a concise explanation of the
emergence of such a
new class of
observable abrupt-change phenomena
we have to point out, first of all, that in the
deep-triple-well regime of the model with $N=2$
the level spacings
$\sqrt{c(\alpha,\beta)} \sim E_0^{\rm (cental)}$
and
$\Omega(\alpha,\beta)
\sim E_0^{\rm (even/odd)}-V(\sqrt{\alpha^2+\beta^2})$
are both of the order of magnitude
${\cal O}(\alpha^2)$, i.e.,
subdominant but still {\em large} and practically
excitation-independent.
For this reason, also
the separate excited states can still be well identified
experimentally.
In the language of the
relocalization theory, its above-outlined
ground-state version
admits an immediate and meaningful upgrade,
therefore. In particular,
the ALC-related Eq.~(\ref{grstsr}) describing
the instant of bifurcation between the central and off-central
localization of the probability density
$\varrho(x)=\psi^*(x)\psi(x)$
becomes complemented by an analogous
excited-state-matching extension
which involves {\em any pair\,} of the bound states lying
in the low-lying spectrum of the system,
 \be
  E_m^{\rm (even/odd)}(\alpha,\beta)=
 E_n^{\rm (central)}(\alpha,\beta)
 + {\rm small \ corrections}\,,
 \ \ \ m,n=0,1,\ldots,M_{\max}\,.
 \label{listsr}
 \ee
At every pair of the main quantum numbers $m$ and $n$
the
search for the solutions remains straightforward, albeit just numerical.

\begin{table}[h]
\caption{The localization of bifurcations
$E_m^{\rm (even/odd)}(\alpha,\beta)=E_n^{(central)}(\alpha,\beta)$
as defined by Eq.~(\ref{listsr}).
The search is performed
at a fixed $\alpha=4\ $ and with a variable $\beta=\mu\,\alpha\ $
or rather $\mu=\sqrt{2+\delta}\,$ with a small $\delta$.
The critical ALC values $\,\delta=\delta(m,n)\,$
were found numerically.}
\vspace{0.5cm}
 \label{owe}
\centering {\small
\begin{tabular}{||c||c|c||}
\hline \hline
  \multicolumn{1}{||c||}{$\delta(m,n)$}
    &m&n
     \\
 \hline
 \hline
    -0.00262&1&3\\
    -0.00261&0&1\\
    \hline
        0.00260&0&0\\
  0.00261&1&2\\
    \hline
  0.00781&1&1\\
  0.00783&2&3\\
    \hline
  0.01299&1&0\\
  0.01302&2&2\\
    \hline
    0.01818&2&1\\
  0.01823&3&3\\
    \hline
  0.02332&2&0\\
  0.02338&3&2\\
   \hline \hline
\end{tabular}}
\end{table}

Table \ref{owe}
offers an illustrative sample of the
corresponding search for the ALC coincidences
obtained as solutions of Eq.~(\ref{listsr})
at $\alpha=4$.
In the Table we consider a variable $\beta=\mu\,\alpha\ $
where the ratio $\mu=\sqrt{2+\delta}\,$
or, more precisely, the value of $\delta$ is a variable parameter.
What one immediately notices is the smallness of
the
deviations
$\delta=\delta(m,n)$ of the numerically evaluated
critical ratios $\mu^2=\beta^2/\alpha^2$
from their unique asymptotic value $\mu_\infty^2={2}$.
Theoretically, this smallness is easily explained
by the proportionality of
the off-central energy values to the local bottom
$V(\sqrt{\alpha^2+\beta^2})$ of the potential. The
decrease of this minimum is dictated by its
very large component
$ -\beta^6/2$ so that the shift
$\delta$ itself
must remain small.
From the point of view of experimentalists,
the high sensitivity  of the
process of the
relocalization
of the probability density
to a parameter
should be interpreted as an abrupt change
of the topology of the system
near a critical $\delta=\delta(m,n)$.
Hence,
it makes sense to speak about
the phenomenon of a quantum relocalization catastrophe
even when the excited states are concerned.

The second striking
feature of the excitation-dependent
relocalization catastrophes which occur at the
$m-$ and $n-$ dependent parameters
$\mu^2=\mu_\infty^2+\delta(m,n)$
may be seen, in Table \ref{owe}, in their approximate
pairwise degeneracy,
with $\delta(1,3)\approx \delta(0,1)$, etc.
Again, the theoretical explanation of the
phenomenon remains straightforward.
The process of the matching of the separate levels
as prescribed by Eq.~(\ref{listsr}) is basically controlled by the
decrease of their off-central partners which is,
roughly speaking, proportional to the increase of $\mu$.
Once we
use the asymptotically correct parameter $\mu=\mu_\infty$
and evaluate the asymptotically correct spring constants
(i.e., level-distances)
$\sqrt{c_\infty}=48$ and $\Omega_\infty=96$, we reveal the
proportionality in disguise,
$\Omega_\infty= 2\,\sqrt{c_\infty}$.

This would imply the exact pairwise degeneracies of the shifts.
Such a phenomenon (reflecting the low-degree polynomial
nature of the interaction)
might deserve an experimental
verification and/or simulation.
In the context of pure theory it is really remarkable that
in our illustrative Table~\ref{owe} such a degeneracy is only
very weakly broken by the nonlinearity of the equations.
We can deduce that even after
the consequent
next-order inclusion of the nonlinearity
in both of our two topologically different dynamical regimes
we will still have,
with good precision, the commensurate level spacings
$\Omega(\alpha,\beta) \approx 2\,\sqrt{c(\alpha,\beta)}$.

\section{Summary\label{soummar}}

In technical terms, our present message has two parts. In one, we
turned attention to the excited states and we strengthened the claim
of paper~I that the mathematical complications introduced by the use
of quantum dynamics still remain surmountable. Secondly, we amended
and completed the results of paper I by an outline of a perturbative
inclusion of the ``missing'', parity-violating components in the
interactions. In this manner we extended the class of the admissible
and tractable models to the whole set of the slightly asymmetric
confining Arnold's potentials. One can say that in spite of the
existence of multiple questions which are still open
(see Appendices A and B for their sample), the overall
picture of the ALC phenomenon formulated in the language of the
low-lying quantum bound states supported by Arnold's potentials is
mathematically consistent.

It is worth adding that the
latter result
was based on a tacit assumption of the smallness
of the antisymmetric part of the potential.
Such an assumption
has two aspects. On positive side it has been shown
useful and efficient in the vicinity of the (approximate)
complete, $(N+1)-$tuple ALC degeneracy where
the branching of the unfolding scenarios becomes maximally
sensitive to the tiniest changes of the parameters.
On negative side, the perturbative tractability of
the other, less extreme ALC degeneracies
were not discussed and their analysis
remains an
open question.
Far from the exceptional dynamical regime
of a maximal ALC degeneracy, unfortunately,
our present perturbation-approximation
technique might fail and
would require an independent study.
This is a serious uncertainty which could
lead to some {\it a posteriori\,}
limitations of our constructive
perturbation approach in applications.

The currently
missing exhaustive classification of
all of the alternative
QRC re-arrangement scenarios might, indeed,
require the use of some alternative,
non-perturbative methods in the future.
At present we must admit that in
such a case it may happen that one
would have to sacrifice the simplicity of the picture.
Perhaps, an entirely different
treatment of the ALC phenomena
will be needed for the quantum Arnold's
potentials characterized
by an extreme spatial asymmetry.

\newpage

\section*{Appendix A: Physics behind the Arnold's potentials: Open questions
and further research}

Naturally, the complexity of the implementation of the basic ideas
of our present paper will increase with the growth of $N$. Via
several tentative preliminary approximative calculations we revealed
that, fortunately, the rate of this increase remains reasonable.
First of all, the proofs of the existence of the topology-changing
quantum catastrophes as well as of the localization of the
corresponding ALC/QRC instants do not change too much with the
growth of $N$. Secondly, the use of any commercially available
symbolic manipulation software appears remarkably suitable for the
purpose. Last but not least, the results of the computer-assisted
algebraic manipulations still remain comparatively compact and
transparent, in spite of being less suitable for the presentation in
print, mainly due to their increasing length. A modified strategy of
their analysis as well as presentation will be needed, in
particular, after the present turn of attention to the asymmetric
potentials $V(x) \neq V(-x)$ and to the low-lying excited states.

Technically, the analysis and predictions proved perceivably less
difficult in the context of paper~I where it has been
conjectured and demonstrated that the Thom's classification of
catastrophes (meaning the abrupt changes of a system's equilibria
after a very small change of its parameters) could consistently be
paralleled in quantum world. The basic idea was of a topological
nature, and the construction task has been facilitated by the
assumption that the parities of the benchmark potentials were even
[cf. their sample in Eq.~(\ref{spefa})]. Furthermore, attention was
restricted to a subfamily of $V(x)$s with the shapes characterized
by the multiple pronounced and deep minima separated by the multiple
high and thick barriers.

A persuasive and systematic qualitative classification has been then
obtained thanks to the acceptance of several auxiliary technical
constraints. The main one can be seen in the attention paid, more or
less exclusively, to the ``strong-coupling'' dynamical regime. In
its framework, both the necessary spectral analysis and the
subsequent predictions of the observable topological effects
characterized, first of all, by the abrupt relocalizations of the
probability densities appeared simplified, feasible and, in
principle, detectable in the laboratory.

In paper I the main observation was that what remains uninfluenced
by the ubiquitous tunneling are the topologically non-equivalent
probability densities $\varrho(x)$. The first nontrivial sample of
such a QRC scenario proved provided by the butterfly potential model
with $N=2$ barriers. In this system  strictly two topologically
non-equivalent equilibria were identified, with either the centrally
dominated or the off-centrally dominated probability density
$\varrho(x)$. On this background it has been claimed that at least
some of the subsequent studies of quantum ALC-related phenomena
might still be based on Schr\"{o}dinger equations with analytic and
polynomial interactions. Soon, the latter expectations were
confirmed by an extension of the analysis to several more or less
realistic two-dimensional \cite{Znojil2d} and three-dimensional
\cite{Znojil3d} descendants of the Thom's one-dimensional cusp and
butterfly potentials of Eq.~(\ref{spefa}). In this context, our
present extension of the scope of the analysis to the spatially
asymmetric Arnold's potentials and to the identification of the ALC
phenomena in the excited quantum states appeared well
phenomenologically motivated, indeed.

In the earlier attempts at the continuation and extension of the
project we were, originally, not too successful. We tried to use the
non-perturbative methods and we wished to extend the theory to the
general asymmetric Arnold's potentials, but we failed. The crisis
had only been overcome when we imagined that the bifurcation of the
maximally degenerate ALC equilibrium could rather straightforwardly
be still described via a restricted, controllably small modification
of the potential. This proved to be one of the key ideas which
inspired our present study and which is also to be pursued in our
future research.

In a more specific preliminary remark let us mention that the most
important open question and theoretical challenge is twofold. In
both of its forms one will encounter a generalization of the theory
based on a relaxation of our present ``traditional'' assumption of
having the Arnold's potentials {\em real}, i.e., of having the
corresponding Hamiltonians {\em self-adjoint}. The two forms of such
a relaxation of constraints may be seen as motivated by the recent
growth of popularity of the so called pseudo-Hermitian reformulation
of the model-building strategy in quantum mechanics (see, e.g.,
reviews \cite{Geyer} or \cite{ali}).

From the point of view of our future ``quantum-catastrophic''
studies and projects, we will have to find a contact with both of
the two conceptually different implementations of the latter
reformulation of quantum theory dealing, respectively, with the so
called open quantum systems (OQS, a monograph \cite{Nimrod} might be
cited in this case) and with the (conceptually very different)
closed quantum systems (CQS) as considered in reviews \cite{Geyer}
and \cite{ali} or \cite{Carl,book}. Indeed, from the different
perspectives, both of these OQS and CQS approaches share the
interest and emphasis put upon the concept of the Kato's exceptional
points. At the same time, we may only repeat that the ALC-EP
correspondence will always play a key role in any future complete
quantum theory of catastrophes.

\section*{Appendix B: Several open mathematical questions
behind the Arnold's potentials: Semiclassical approximation, etc}

One of the explanations of the success and popularity of the Thom's
classification of elementary non-quantum catastrophes lies in its
simplicity. Indeed, the emphasis upon the qualitative aspects of the
classical evolution patterns made his theory universal. In parallel,
the underlying intuitive perception of the concept of the
equilibrium rendered many of its applications straightforward. {\it
A priori\,} one would expect that both of these merits of the Thom's
theory (viz., its simplicity and a straightforward applicability)
must necessarily be lost after quantization.

In this sense, our present paper can be read as advocating an
opposite opinion. In a way co-supported by the results of
Ref.~\cite{ArnoldI} we may claim that up to the trivial single-well
and single-barrier exceptions (where any observable change remains
smooth, due to the tunneling in the second case), all of the
multi-parametric confining Arnold's potentials prove able to serve
as nontrivial benchmark realizations of the ALC-related quantum
relocalization processes. In a way caused by a tiny change of the
parameters these processes were shown to exhibit all of the
characteristic features (like, in particular, the abruptness) of the
popular classical catastrophes.


After the latter conclusion one has to ask the natural question
concerning the dependence of the ALC phenomenon on a hypothetical
strengthening $V(x) \to g\,V(x)$ of the Arnold's potential. The
question is legal because in Eq.~(\ref{pefa}) which defines the
potential the leading-power coefficient is equal to one. Indeed,
once we fixed the units (such that ${\hbar^2}/({2\mu})=1$) we lost
the variability of the parameter which could control the
semiclassical limit of the model. In other words, we fixed the arrow
of our inspiration (from classical to quantized) and we lost the
opportunity of studying the quantum-classical correspondence of the
systems under consideration. In the opposite direction of moving
from quantum to classical, the idea of reduction of the quantum
picture to its classical limit could find its implementation, in the
future research, in a way based on the variational \cite{Turbiner},
semiclassical \cite{Simon} or even some brutally numerical
treatments of the initial quantum system.

Technically, the goal could be achieved when one abbreviates
${\hbar^2}/({2\mu})= {\Lambda}^2$ and keeps the latter parameter
variable. One of the consequences would be that Schr\"{o}dinger
equation
 \be
 \left[- {\Lambda}^2\,\frac{d^2}{dx^2} +
 V^{}(x)
 \right ] \psi_n (x) =E_n\,
 \psi_n (x)
  \,, \ \ \ \ \
\,,\ \ \ \ \ \
 n=0, 1, \ldots\,
 \label{kseh}
 \ee
becomes more flexible. What remains unchanged is that near one of
the deep minima the potential can still be replaced by its
harmonic-oscillator approximation. This would yield the approximate
low-lying spectrum in which the $\Lambda=1$ energies
 $
 E_{n} \sim
  V(x_{\min})+ \,(2n+1){\omega}
 $
would be replaced by their rescaled $\Lambda\neq 1$ descendants
 $
 E_{n}(\Lambda) \sim
  V(x_{\min})+ \Lambda\,(2n+1){\omega}
 $. The variable $\Lambda$
interconnects the large-mass and semi-classical limits as well as,
alternatively the small-mass and ultra-quantum dynamical regimes
(cf. also \cite{Turbiner} in this respect).

The classical catastrophe theory reemerges in the semiclassical
limit in which all of the low-lying levels would converge to the
minimum of the potential \cite{Simon}. Vice versa, the quantum
effects are enhanced at large ${\Lambda}$ while becoming divergent
in the vanishing-mass limit $\mu \to 0$.

\end{document}